\documentclass[11pt, a4paper,reqno]{amsart}
\usepackage{amsfonts}
\usepackage{amssymb}
\usepackage[dvips]{graphics}
\usepackage{epsfig}
\usepackage{euscript}
\usepackage{color}

\usepackage{fullpage,hyperref}

\usepackage{fancyvrb}

 \newtheorem{thm}{Theorem}[section]
 \newtheorem{cor}[thm]{Corollary}
 \newtheorem{lemma}[thm]{Lemma}
 \newtheorem{prop}[thm]{Proposition}
 \theoremstyle{definition}

 \newtheorem{ex}{Example}
 
 \numberwithin{equation}{section}

\newcommand{\caA}{{\mathcal A}}

\newcommand{\caF}{{\mathcal F}}

\newcommand{\caH}{{\mathcal H}}
\newcommand{\caI}{{\mathcal I}}

\newcommand{\caL}{{\mathcal L}}

\newcommand{\caO}{{\mathcal O}}

\newcommand{\bbC}{{\mathbb C}}

\newcommand{\bbN}{{\mathbb N}}

\newcommand{\bbR}{{\mathbb R}}

\newcommand{\bbZ}{{\mathbb Z}}

\newcommand{\iu}{\mathrm{i}}

\newcommand{\str}{^{*}}
\newcommand{\ep}[1]{\mathrm{e}^{#1}}

\newcommand{\dd}{\,\mathrm{d}}
\newcommand{\tr}{\mathrm{tr}}
\newcommand{\Tr}{\mathrm{Tr}}

\newcommand{\norm}[1]{\left\Vert #1 \right\Vert}

\newcommand{\dist}{\mathrm{dist}}

\newcommand{\fatlattice}{\bbZ_{\scriptscriptstyle{(\caO(L^{-\infty}))}}}

\newcommand{\be}{\begin{equation}}
\newcommand{\ee}{\end{equation}}
\newcommand{\bea}{\begin{eqnarray}}
\newcommand{\eea}{\end{eqnarray}}
\newcommand{\beann}{\begin{eqnarray*}}
\newcommand{\eeann}{\end{eqnarray*}}

\newcommand{\error}{\caO(L^{-\infty})}

\newcommand{\V}{{}}

\newcommand{\var}{\mathrm{Var}}

\newcommand{\vsmall}{  \caO(L^{-\infty})}  

\newcommand{\K}{K}

\begin{document}

\title{A many-body index for quantum charge transport}

\author{Sven Bachmann}
\address{Department of Mathematics \\ The University of British Columbia \\ Vancouver, BC V6T 1Z2 \\ Canada}
\email{sbach@math.ubc.ca}

\author{Alex Bols}
\address{ Instituut Theoretische Fysica, KU Leuven  \\
3001 Leuven  \\ Belgium }
\email{alexander.bols@kuleuven.be}

\author{Wojciech De Roeck}
\address{ Instituut Theoretische Fysica, KULeuven  \\
3001 Leuven  \\ Belgium }
\email{wojciech.deroeck@kuleuven.be}

\author{Martin Fraas}
\address{Department of Mathematics \\ Virginia Tech \\ Blacksburg, VA 24061-0123 \\ USA}
\email{fraas@vt.edu}

\date{\today }

\begin{abstract}
We propose an index for pairs of a unitary map and a clustering state on many-body quantum systems. We require the map to conserve an integer-valued charge and to leave the state, e.g.\ a gapped ground state, invariant. This index is integer-valued and stable under perturbations. In general, the index measures the charge transport across a fiducial line.  We show that it reduces to (i)~an index of projections in the case of non-interacting fermions, (ii)~the charge density for translational invariant systems, and (iii)~the quantum Hall conductance in the two-dimensional setting without any additional symmetry. Example (ii) recovers the Lieb-Schultz-Mattis theorem, and (iii) provides a new and  short proof of quantization of Hall conductance in interacting many-body systems. 
\end{abstract}

\maketitle


\section{Introduction}\label{sec:intro}

By definition, two gapped ground states belong to the same quantum phase if the Hamiltonians can be smoothly connected while keeping the gap open \cite{HastingsWen}. Discrete-valued indices, which cannot change continuously, are therefore invariants of quantum phases:  They show universality and stability. 
There has been substantial progress towards classification of quantum phases for both non-interacting systems \cite{heinzner2005symmetry,KitaevTable,schnyder2008classification,ProdanBook}, one-dimensional interacting systems~\cite{perez2008string,chen2011classification,pollmann2012symmetry,gross2012index,bachmann2014gapped,cirac2017matrix,Watanabe,Ogata}, and recently for higher dimensional interacting systems as well~\cite{RuyClassificationManyBody}.

 In this paper, we use a version of Laughlin's flux threading argument to construct a microscopic, additive many-body index. This approach is very close to the works \cite{oshikawa2000commensurability,lu2017filling,matsugatani2018universal} that unify two indices: the filling factor and the Hall conductance. Our expression of the index has the meaning of a charge transported across a fiducial line. It is directly inspired by the well-known index of projections  introduced in~\cite{ASS90} and it reduces to it in the case of non-interacting fermions.

\subsection{The index}\label{sub:IntroIndex}

Before giving the precise setting, we present the index in a loose language that assumes familiarity with many-body lattice systems. We consider a discrete $d$-torus $\Lambda=\Lambda_L$ consisting of $L^d$ sites. We focus on a bipartition of the torus given by $\Lambda = \Gamma \cup \Gamma^c$, where $\Gamma = \{0\leq x_1 < \frac{L}{2}\}$, see Figure~\ref{fig:sets_torus}.
For now, we refrain from indicating this explicitly, but all {assumptions} are meant up to an error that vanishes as $L\to\infty$.  
The index requires three basic ingredients: 
\begin{enumerate}
\item[$1)$] A family of local \textbf{charge} operators $\{Q_x: x\in\Lambda\}$ with integer spectrum. Let us write $Q_X = \sum_{x\in X}Q_x$ for the charge in region $X$.   A special role is played by  $Q := Q_{\Gamma}$.
\item[$2)$] A (quasi-)local \textbf{unitary} $U$ that conserves charge {locally}. In particular,  the operator $U\str Q U - Q$ is supported in a neighbourhood of the boundary  $\partial \Gamma=\partial_- \cup\partial_+$, see again Figure~\ref{fig:sets_torus}. We can hence split  $U\str Q U - Q=T_-+T_+$.   The operators $T_\pm$ represent the charge transported across the lines $\partial_{\pm}$ by the unitary $U$.
\item[$3)$] A $U$-invariant \textbf{state} $\Omega$ that is exponentially clustering in the $x_1$-direction. The state $\Omega$ has {local charge fluctuations} in the sense that there are operators $K_-$ and $K_+$ acting in a neighbourhood of $\partial_-$ and $\partial_+$ such that $\Omega$ is an eigenstate of $Q - \K_{-}-K_{+}$.
\end{enumerate}
\begin{figure}
\centering
\includegraphics[width = 0.3\textwidth]{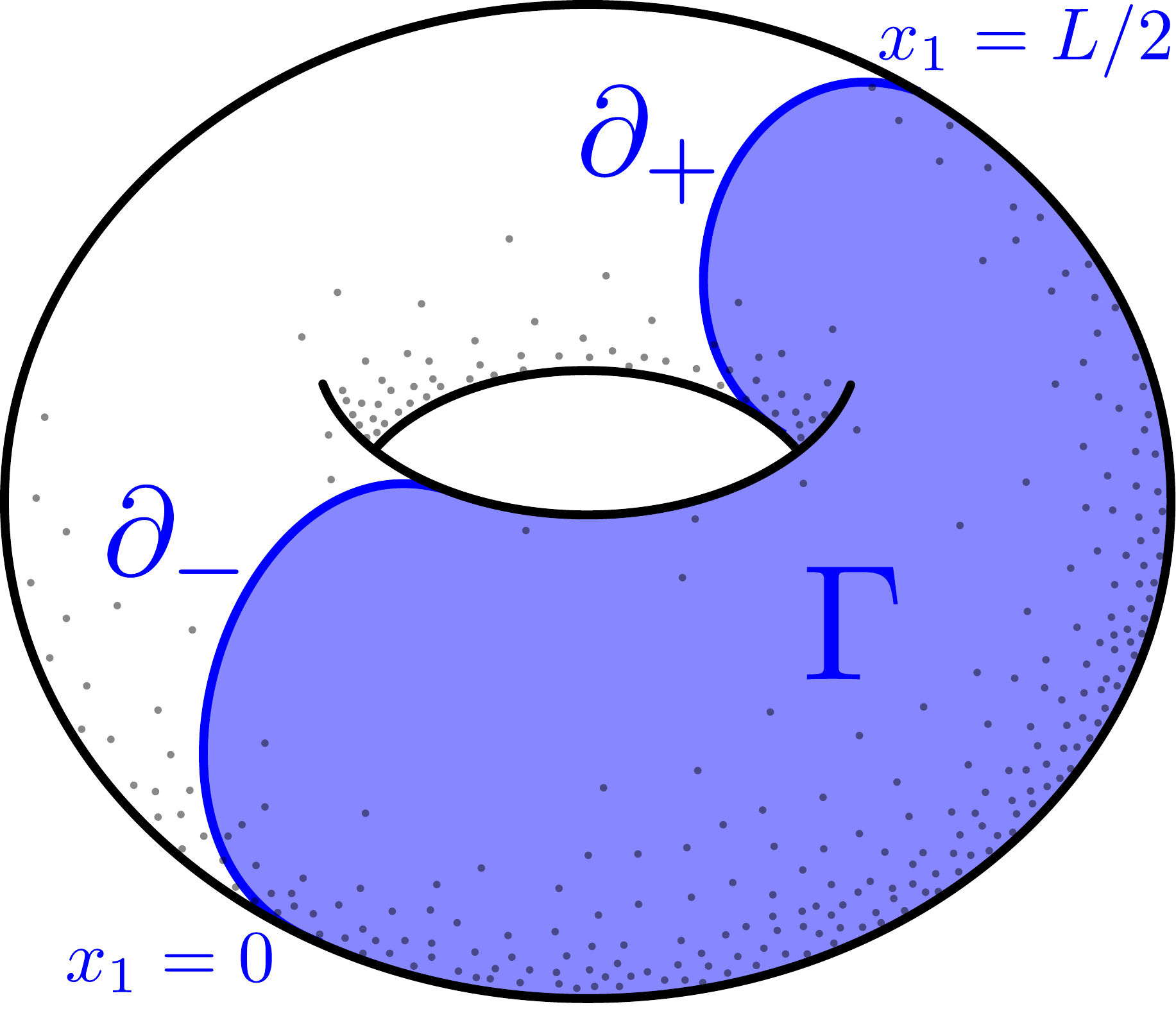}
\caption{The set $\Gamma$ and its boundaries}\label{fig:sets_torus}
\end{figure}

This setup is familiar except for the requirement of local charge fluctuations. That assumption is however satisfied if  $\Omega$ is a gapped ground state of a charge conserving local Hamiltonian. 

{Since $\Omega$ is $U$-invariant, we have obviously $\left\langle\Omega,(T_-+T_+) \Omega\right\rangle =0$, indicating that it is only interesting to consider the lines $\partial_\pm$ separately. }
The splitting of $U\str Q U - Q=T_-+T_+$ is not unique because $(T_-+\nu)+(T_+-\nu)$ is a valid splitting as well, for any $\nu \in \bbR$. We hence shall choose a splitting, namely a definition of~$T_{\pm}$, but we constrain the choice by requiring $\ep{2\pi\iu (Q+T_-)}=1$; we will later explain why this is always possible. In all examples that we consider, there is a physically natural choice for $T_{\pm}$, satisfying in particular the above requirement.
Then, the main claim of this work is that the mean charge transported across a single line is quantized: 
\begin{equation*}
\left\langle\Omega, T_- \Omega\right\rangle \, \in \, \bbZ.
\end{equation*}
%
%

\subsection{Examples}\label{sec: examples}

We now illustrate the theorem with five examples. A detailed discussion of these examples is in Section~\ref{sec:Examples}. 

\begin{ex}[Index of projections] We consider non-interacting lattice fermions in $d=1$. We take all the many-body objects in the index formula as fermionic second quantization of the corresponding objects on the Hilbert space $l^2(\Lambda)$.
The charge operator $Q$ is the second quantization\footnote{The symbols $\Gamma$ and $\dd \Gamma$ (not to be confused with the spatial region $\Gamma$) are functors mapping one-particle operators to many-body operators, see e.g.\ \cite{reed1975methods}} $\dd \Gamma(q)$ of the indicator function $q=1_\Gamma$, the ground state $\Omega$ is the gauge-invariant quasi-free state corresponding to a Fermi projection $p$, and the unitary $U$ is the second quantization $\Gamma(u)$ of a one-particle unitary $u$. Clustering holds whenever the chemical potential lies in a spectral gap or a mobility gap. In the limit $L\to \infty$, $1_\Gamma \to 1_\bbN$ weakly and the many-body index takes the form,
\begin{equation}\label{FermionsConv}
\lim_{L \to \infty} \left\langle\Omega, T_- \Omega\right\rangle = \tr(u\str k u - k),\qquad k = p 1_\bbN  p. 
\end{equation}
with $u,p$ appropriate infinite-volume analogues for the objects above.
The operator $k$ can be deformed to a projector $\chi(k)$ without changing the value of the right hand side, which we can hence call  $\mathrm{Index}( \chi(k); u)$ in the notation of~\cite{ASSIndex}, as a particular case of a Fredholm index \cite{fredholm1903}.

\end{ex}

\begin{ex}[Lieb-Schultz-Mattis theorem] Consider now an interacting system that is translation-invariant in the $x_1$-direction, with $U$ implementing this translation $x_1 \mapsto x_1+1$. The index then takes the form
$$
-\left\langle\Omega, T_- \Omega\right\rangle = \langle\Omega, Q_{[0]} \Omega\rangle= \frac{1}{L}\langle\Omega, Q_{\Lambda} \Omega\rangle,
$$
where $[0]=\{x_1=0\}$, i.e.\ the charge in the $x_1=0$ hyperplane. 
The theorem then states that if $\Omega$ is clustering -- as is the case for a unique gapped ground state -- then the average charge per transverse layer is an integer. This is a version obtained in \cite{oshikawa2000commensurability} of the celebrated result of~\cite{LSM}, generalized in~\cite{LSM2,HastingsLSM,BrunoLSM,chen2011classification}. Often, this theorem is discussed in a setting with half-odd integer spins and $\mathrm{SU}(2)$-invariant Hamiltonian, with $Q_x=S^{(3)}_x+\tfrac{1}{2}$ and $S^{(\alpha)}_x,\alpha=1,2,3$ the spin operators.   In that case, the `charge density' in singlet states is fractional by construction if the total number of sites is odd, so the theorem actually rules out an $\mathrm{SU}(2)$-symmetric non-degenerate gapped ground state, see also~\cite{oshikawa2000commensurability}. Recent work \cite{Watanabe,OgataTasaki} has shown that this setting can be generalized even further, replacing $\mathrm{SU}(2)$ by a discrete symmetry.
\end{ex} 

\begin{ex}[Quantum Hall effect]
Here, $\Omega$ is the ground state of a many-body gapped Hamiltonian in $d=2$ dimensions. The unitary $U$ models the threading of one unit of flux pointing in the $x_1$ direction, see Figure~\ref{fig:torus}.  Standard considerations show that the transported charge across the line $x_1=0$, i.e.\   $\left\langle\Omega, T_- \Omega\right\rangle$, equals $2\pi$ times the Hall conductance. Therefore, our result yields an alternative and independent proof of the quantization of Hall conductance. It does not rely on a Chern number argument and does not require a spectral gap for all fluxes.

\begin{figure}
\centering
\includegraphics[width = 0.3\textwidth]{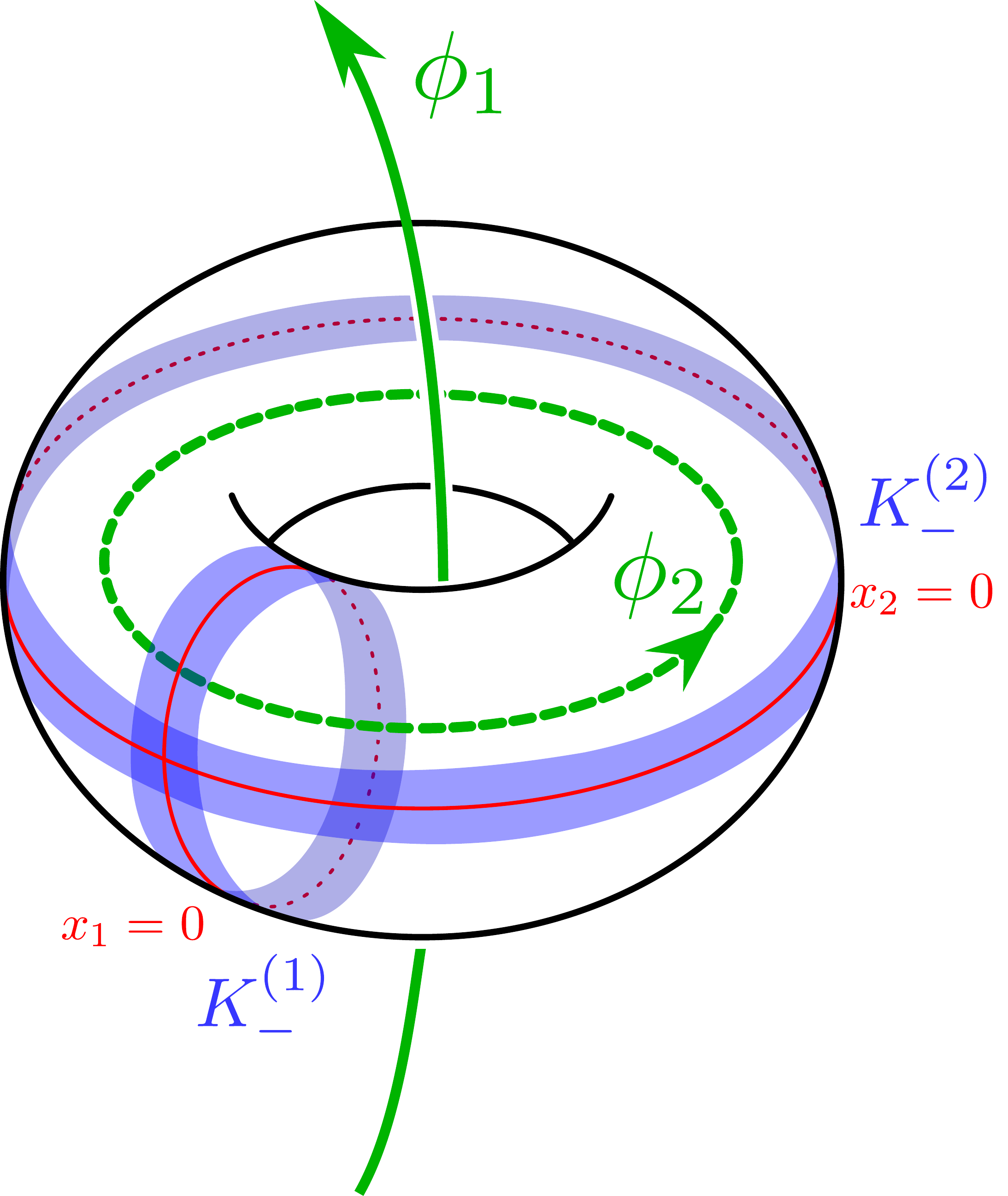}
\caption{Threading fluxes $\phi_1,\phi_2$ in the torus.}\label{fig:torus}
\end{figure}

In the non-interacting setting, the Hall conductance was related to an index in \cite{ASS90}, and in \cite{AizenmanGraf} in the case with disorder. Quantization of Hall conductance with interactions was proved in \cite{HastingsMichalakis}, see also~\cite{OurQHE} for a simplified proof and~\cite{giuliani2017universality} for the case of weakly interacting systems. Earlier mathematical approaches are described in~\cite{frohlich1997classification,BellissardSchuba}.
\end{ex}

\begin{ex}[Thouless pump]
This is a generalization of the previous example, in that $U$ can now refer to a an arbitrary  process preserving $\Omega$, not necessarily flux threading. To be concrete, $U$ usually corresponds to the unitary operator implementing the parallel transport associated with an adiabatic cycle of gapped Hamiltonians. The charge transported in a cycle is then quantized. The Thouless pump was first discussed in a non-interacting, spatially periodic setting in~\cite{Thouless83}, and generalized to disordered and interacting systems in~\cite{niu1984quantised}. It was related to the scattering approach of quantum charge transport in~\cite{braunlich2010equivalence}, where the B\"uttiker-Pr\^etre-Thomas formula~\cite{BPT} is shown to take the shape of a winding number.
\end{ex}

\begin{ex}[Bloch's Theorem]
Finally, we consider the family $U_t = \ep{-\iu tH}$, namely the propagator for some charge conserving local Hamiltonian $H$, and $\Omega$ is its gapped non-degenerate ground state. The index is then $t$ times the stationary ground state current. Since the index is integer, and $t$ can be continuously reduced to $t=0$, we conclude that the current across a fiducial line vanishes in the ground state, which is the content of Bloch's theorem \cite{Bohm,watanabeBloch}.
\end{ex}


\section{The index theorem}\label{sec:Index}

We first describe the technical setup in Section \ref{sec:setting}. Section \ref{sec: index theorem} contains the main result. We comment in Section~\ref{sec:comments} on various aspects of the result, in particular on the stability if the index and its thermodynamic limit.
We then elaborate on two special cases: first when  $\Omega$ is a gapped ground state (Section \ref{sec: case of gapped}), and second when the unitary $U$ is continuously connected to the identity (Section \ref{sec: trivial unitaries}). 

\subsection{Preliminaries}\label{sec:setting}

\subsubsection{Spatial setup} \label{sec: spatiall setup}
We consider a discrete $d$-torus  $\Lambda=\Lambda_L =\bbZ^d_L$ where $\bbZ_L=\bbZ/(L\bbZ)$ is identified with $\{-L/2+1,\ldots, L/2-1,L/2\}$  (we choose $L$ even), so that the cardinality is $|\Lambda|=L^d$.  
We define $\Gamma$ to be the half-torus $\Gamma=\{ 0\leq x_1 < L/2\}$. For a region $X$ and $r \in \bbN$, we denote
\begin{equation*}
X^r = \{x\in X: \mathrm{dist}(x,X)\leq r\},
\end{equation*}
where $\mathrm{dist}(\cdot,\cdot)$ is the graph distance on $\Lambda$. We use this often to fatten the boundaries $\partial_- = \{x_1 = 0\}$, $\partial_+ = \{x_1=L/2-1\}$ of $\Gamma$. Finally, we define the strips 
$$
S_{-} := \{ -\tfrac{L}{8} \leq x_1\leq \tfrac{L}{8}\},\qquad 
S_m := \{ \tfrac{L}{8} < x_1< \tfrac{L}{2}-\tfrac{L}{8}\},\qquad
S_{+} := \{ -\tfrac{L}{8} \leq x_1-\tfrac{L}{2}\leq \tfrac{L}{8}\},
$$
see Figure~\ref{fig: torus}.

We will treat spin systems and fermionic lattice systems on the same footing. We first give a brief account of both.
 
\subsubsection{Spin systems}
 Each site $x\in\Lambda$ carries a finite dimensional Hilbert space $\caH_x\equiv\bbC^n$ and $\caH=\caH_\Lambda = \otimes_{x\in \Lambda}\caH_x$ is the total Hilbert space.  The algebra of observables is $\caA= \caL(\caH)$ and we consider the subalgebras $\caA_X, X\subset \Lambda$ of observables of the form $O=O_X \otimes 1_{\Lambda\setminus X}$. We say that $A$ is supported in $X$ whenever $A\in\caA_X$.

\begin{figure}
\centering
\includegraphics[width = 0.25\textwidth]{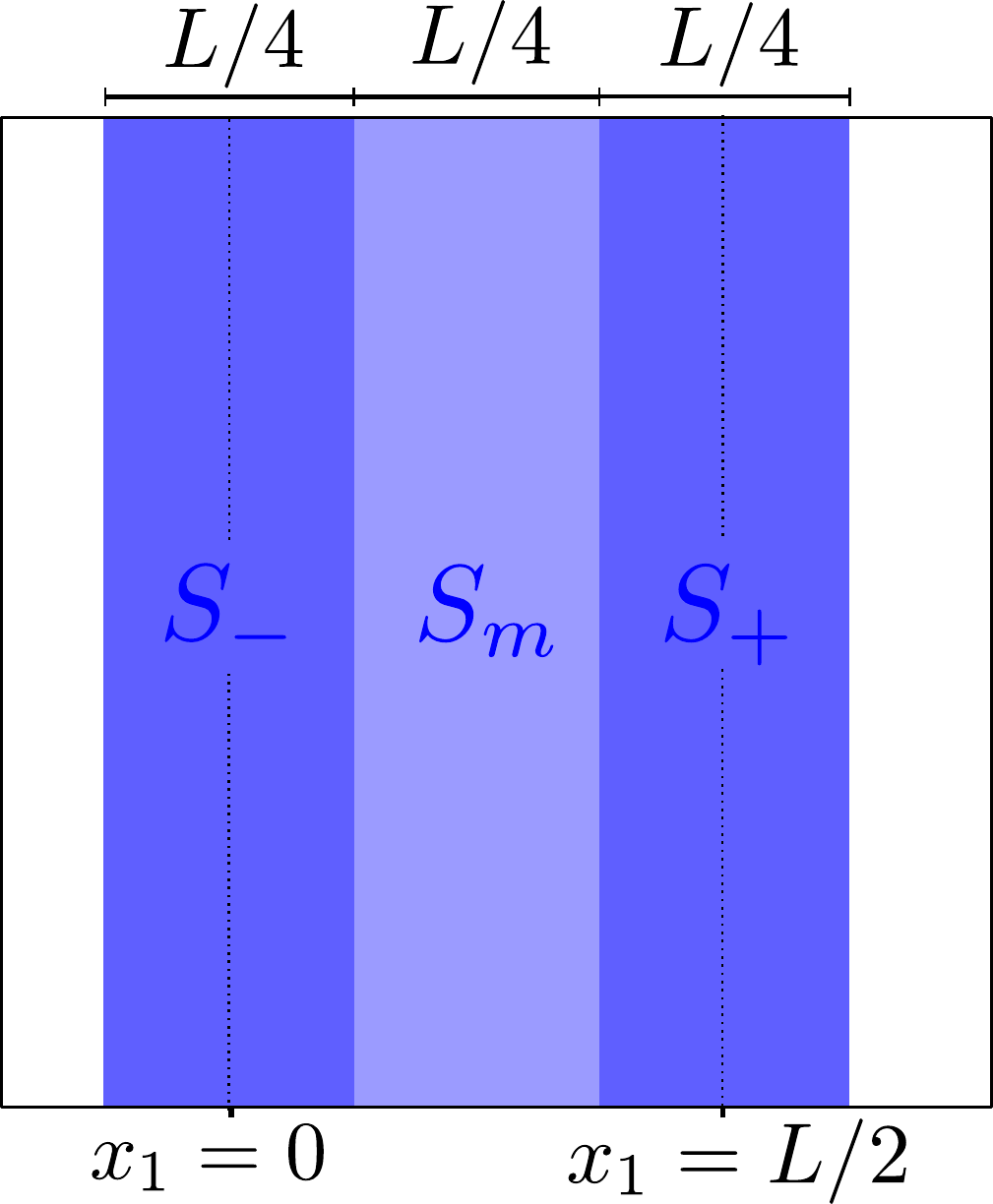}
\caption{The sets $S_{-}, S_m$ and $S_+$}
\label{fig: torus}
\end{figure}
 
 \subsubsection{Fermions}

The fermionic Fock space $\caH=\caH_{\Lambda}$ is the antisymmetric second quantization of the one-particle space $l^2(\Lambda,\bbC^n)$. There is a preferred basis in $\bbC^n$ labelled by $\sigma$ (as an example, one can think of the $z$-spin number). 
The algebra of canonical anticommutation relations $\widetilde{\caA}=\caL(\caH)$ is generated by the identity and the fermionic creation/annihilation operators $\{c_{x,\sigma},c^*_{x,\sigma}: x\in\Lambda,\sigma = 1,\ldots,n\}$ acting on $\caH$, which satisfy: 
$$
\{c_{x,\sigma},c^*_{x',\sigma'} \}=\delta_{x,x'} \delta_{\sigma,\sigma'}, \qquad   \{c^{\sharp}_{x,\sigma},c^{\sharp}_{x',\sigma'} \}= 0,
$$
where $\{A,B\}=AB+BA$ and $c^\sharp$ can be either $c$ or $c^*$. 
For any spatial set $X$,  $\widetilde{\caA}_X$ is the algebra generated by $\{c_{x,\sigma},c^*_{x,\sigma}: x\in X,\sigma = 1,\ldots,n\}$. We define the even subalgebra $\caA_X\subset \widetilde{\caA}_X$ generated by all monomials with an even number of  $c^\sharp$ operators. Alternatively, $\caA_X$ is the set of operators in  $\widetilde{\caA}_X$ that commute with the fermionic parity $\caF=\prod_{x,\sigma} (-1)^{n_{x,\sigma}}$, where $n_{x,\sigma}=  c^*_{x,\sigma} c_{x,\sigma}$. 
As for spin systems, an even element $A\in \caA_X$ is said to be supported in $X$. The restriction to even operators allows us to treat fermions and spins on the same footing, because it 
enforces the locality property of the observable algebra that we discuss in the following subsection.

\subsubsection{Locality} From now on, we do not distinguish between spin systems and fermions any more. For both, we first of all have the following basic locality property: If $X \cap X'=\emptyset $, then 
$$
  [A,A'] =0 \qquad \text{for any $A\in \caA_X, A' \in \caA_{X'}$}.
$$
Secondly, the Heisenberg dynamics of an observable $A\in \caA_X$, generated by a local Hamiltonian, satisfies a Lieb-Robinson bound~\cite{Lieb:1972ts}. These two properties, which are crucial for our results, are the main reasons why $\caA$ is chosen to be the even algebra in the fermionic case, see~\cite{Bru:2017aa,BrunoFermions}.

\subsubsection{Arbitrarily large finite systems}

In this work, we treat all finite $L$ simultaneously. The essential point of our results is that bounds hold for all $L$ with constants $c,C$ that are uniform in $L$. This means, strictly speaking, that we consider a sequence of models indexed by $L$.  For example, consider the notion of clustering for states discussed below.  If this was referring to a single fixed $L$, then the concept would be empty as one can always find $C,c$ such the bound holds for a given $\psi$.  We therefore mean it to refer to a sequence of states $\psi=\psi_L$ with fixed, $L$-independent $c,C$. This will not be repeated at every step.

\subsubsection{States}

A state is a normalized vector $\psi \in \caH, \norm{\psi}=1$.
The variance of $A$ with respect to $\psi$ is
\begin{equation} \label{def: var}
\var_\psi(A) :=  \norm{   (A -  \langle\psi, A \psi\rangle)\psi }^2.
\end{equation}
We say that a state $\psi$ is clustering in the $x_1$-direction, if there are $0<c,C<\infty$ such that
\begin{equation}
\left\vert \langle\psi,AB\psi\rangle - 
\langle\psi, A\psi\rangle \langle\psi, B\psi\rangle\right\vert
\leq C \Vert A\Vert \Vert B \Vert \vert X \vert\,\vert X'\vert\ep{-c d_1(X,X')}
\label{A:clustering}
\end{equation}
for all $A\in\caA_X,B\in\caA_{X'}$, where $$d_1(X,X') := \min\{\vert x_1 - x_1'\vert:x\in X,x'\in X'\}.$$

\subsection{The index theorem}\label{sec: index theorem}

We now phrase the discussion of Section~\ref{sub:IntroIndex} in precise terms.
We assume local charge operators $Q_x=Q_x\str\in\caA_{\{x\}}$ with integer spectrum,
\begin{equation}\mathrm{Spec}(Q_x)\subset\bbZ.
\label{A:IntegerCharge}
\end{equation}
{For any $X\subset\Lambda$, we write $
Q_X= \sum_{x\in X}Q_x$ and we use the abbreviation $Q=Q_{\Gamma}$ since $\Gamma$ plays a central role. 
We consider a unitary $U \in\caA$ and a state $\Omega$ that satisfy the following assumptions.
\begin{enumerate}
\item   The unitary $U$ is an \emph{almost local unitary} in the sense that if $A\in\caA_X$, then for each $n\in\bbN$, there is $R_n \in \caA_{X^n}$ such that 
\begin{equation}\label{DefLocalUnitary}
\frac{\left\Vert U\str A U -R_n \right\Vert}{\Vert A\Vert \vert X \vert} =   \caO(n^{-\infty}).
\end{equation}
\item The unitary $U$ satisfies the following \emph{local charge conservation}:
There are self-adjoint operators  $T_{\pm} \in \caA_{S_{\pm}}$ such that 
\begin{equation}\label{A:Uloc}
\left\Vert   (U\str Q U - Q)
- T_{-}-T_{+}\right\Vert
= \caO(L^{-\infty}),
\end{equation}
expressing that the transported charge $U\str Q U - Q$ is supported near the boundary of~$\Gamma$. This assumption is implied by ~(i) and the global conservation law $[U,Q_\Lambda]=0$. 
\item 
The state $\Omega$ is an \emph{approximate eigenvector} of $U$ in the sense that
\begin{equation}
\label{A:Invariance}
\var_\Omega(U)=\caO(L^{-\infty}).
\end{equation}
\item 
The state $\Omega$ has \emph{local charge fluctuations}: There exist self-adjoint operators $\K_{\pm}\in\caA_{\partial_{\pm}^{L/16}}$ with $\norm{K_{\pm}} \leq C |\Lambda|$ such that
\begin{equation}\label{A:QVar}
\var_\Omega(Q - K_- - K_+)=\caO(L^{-\infty}).
\end{equation}
\item The state $\Omega$ is clustering in $x_1$-direction.
\end{enumerate}
For $0<\epsilon<1/2$, we denote by $\bbZ_{(\epsilon)}$ the fattened lattice $ \bbZ_{(\epsilon)} =\{ z\in\bbR :  \mathrm{dist}(z,\bbZ) \leq \epsilon \}$.
As we prove below, Assumptions~(i,ii) immediately imply that we can actually choose the operators $T_{\pm}$, satisfying \eqref{A:Uloc}, such that 
\begin{equation} \label{eq: choice t}
\ep{2\pi\iu (Q+T_-)}=1+\caO(L^{-\infty}),\qquad \Vert T_{\pm}\Vert\leq C|\Lambda|.
\end{equation}
We fix such a choice and we are now ready to state our main result.
\begin{thm}\label{thm:main}
If Assumptions~\emph{(i)-(v)} above hold, and for $T_-$ satisfying the constraint~(\ref{eq: choice t}), then
\begin{equation*}
\langle\Omega,T_{-}\Omega\rangle   \in \bbZ_{\scriptscriptstyle{(\caO(L^{-\infty}))}} .
\end{equation*}
\end{thm}
\subsection{Comments}\label{sec:comments}
\subsubsection{Constraints on the choice of $T_{\pm}$}\label{sub:nu}
We consider the splitting
\begin{equation}\label{Decomp of Q}
Q=Q_{\Gamma}=Q_-+Q_m+Q_+
\end{equation}
for the charges $Q_{\Gamma \cap S_\alpha}$, with $\alpha=(-,m,+)$,  respectively. 
Let $\tilde T_{\pm}$ be a choice satisfying Assumption~(ii).  We note that
$$
U\str Q U= (Q_-+\tilde T_-)+ Q_m + (Q_++\tilde T_+) +\caO(L^{-\infty})
$$
where the operators  $(Q_-+\tilde T_-),  Q_m, (Q_++\tilde T_+)$ are supported on  $S_{-},  S_m, S_{+}$, respectively, see Section~\ref{sec: spatiall setup}. Since they act on distinct tensor factors, the spectrum of their sum is the sumset of their spectra. As the spectrum of $Q_m$ is integer, and that of $U\str Q U$ as well, it follows by spectral perturbation theory that there is a $\nu \in \bbR$ such that
$$
\mathrm{Spec} (Q_-+\tilde T_-) \in \bbZ_{\scriptscriptstyle{(\caO(L^{-\infty}))}}-\nu,\qquad  \mathrm{Spec} (Q_++\tilde T_+) \in \bbZ_{\scriptscriptstyle{(\caO(L^{-\infty}))}}+\nu.
$$
So we let 
$$
T_{\pm}:= \tilde T_{\pm} \mp \nu,
$$
which ensures indeed that $\ep{2\pi\iu (Q+T_\pm)}=1+\caO(L^{-\infty})$. The possibility of satisfying the bound $\Vert T_{\pm}\Vert \leq C|\Lambda|$ follows from $\Vert Q\Vert, \Vert U^*QU\Vert \leq C|\Lambda|$. Of course, even these constraints define $T_{\pm}$ only up to (not too large) integers. As already pointed out, all examples come with a natural choice. 

\subsubsection{Geometry}
For clarity and simplicity, we have chosen the total volume $\Lambda$ to be a discrete $d$-dimensional torus, but this is not essential. In particular, the topological nature of Theorem~\ref{thm:main} is unrelated to this spatial torus. 
The results would hold just as well in the following setup, and proofs would only require cosmetic modifications.  Let $\Lambda=\Lambda_L$ be a graph of diameter $L$ in graph distance. 
Let it have bounded dimension $d$ in the sense that    $\sup_{x'\in\Lambda}|\{x\in\Lambda : \dist(x,x')<R\} | \leq C R^d$. 
Let there be a set $\Gamma$, replacing the half-torus, such that its boundary $\partial \Gamma=\{x\in\Lambda : \dist(x,\Gamma) \leq 1 \text{ and } \dist(x,\Gamma^c) \leq 1  \}$ can be partioned in two sets $\partial \Gamma=\partial_{-} \cup \partial_+$ such that $\dist(\partial_-,\partial_+) >c L$. The operator $T_-$ is now measuring the charge transport across the boundary $\partial_-$.

\subsubsection{Stability of the index} \label{sec: stability of index}

We consider the stability of the index under continuous deformations of the unitary $U$ and of the state $\Omega$.
To that end, let $[-1,1]\ni s\mapsto U_s,\Omega_s$ be families such that Assumptions (i)-(v) are satisfied for each $s$, with implicit constants featuring in the definition of quasi-locality in $U$ and the bounds $\caO(L^{-\infty})$ that can be chosen uniformly is $s$.  
In particular, the indices $\langle \Omega_s,T_{-}(s)\Omega_s\rangle$ are quantized, pointwise for any $s$. We assume that  $U_s$ and $\Omega_s$ are continuous in $s$ (topology does not matter since these are finite-dimensional vector spaces), and the modulus of continuity may depend on L.

We now show that the index is continuous and hence nearly constant for large $L$, provided the choice of $T_-(s)$ is consistent along~$s$, which can always be done.
 \begin{prop}\label{prop:stability}
There is a choice $T_-(s)$ satisfying Assumptions~\emph{(i)-(v)}, the constraint \eqref{eq: choice t}, and such that $ s\mapsto T_-(s)$ is continuous.  For such a choice of $T_-(s)$, the index is constant up to $\vsmall$:
\begin{equation}\label{Constant Index}
  \sup_{s\in[-1,1]} |\langle \Omega_s,T_{-}(s)\Omega_s\rangle -\langle \Omega_0,T_{-}(0)\Omega_0\rangle|  =\vsmall.
\end{equation}
In particular, the integer closest to $\langle \Omega_s,T_{-}(s)\Omega_s\rangle$ is independent of $s$ for $L$ large enough.
\end{prop}
\begin{proof}
By assumption, the family $s\mapsto U_s\str Q U_s - Q$ is continuous, which implies the existence of a continuous choice for~$T_{-}(s)$, see Section~\ref{sub:nu}. It follows that $s\mapsto \langle \Omega_s,T_{-}(s)\Omega_s\rangle$ is continuous. Since there is $n(s)\in\bbZ$ such that 
\begin{equation*}
\vert \langle \Omega_s,T_{-}(s)\Omega_s\rangle - n(s) \vert = \caO(L^{-\infty}),
\end{equation*}
we conclude that, for any fixed $L$, the function $s\mapsto n(s)$ is locally constant , and hence constant on $[-1,1]$. 
This in turn implies~(\ref{Constant Index}). 
\end{proof}

\subsubsection{Thermodynamic limit}
In this subsection, we reinstate the $L$-dependence. The existence of the thermodynamic limit, $L \to \infty$, for local expectations in the state $\Omega_L$ is independent of all the previous considerations. If we assume it, then we can discuss convergence of the index, and in that case, the index is easily showed to converge provided the observable $T_-$ is sufficiently local. By identifying the torus with a square, the C*-algebras of observables $\caA_L$ for different values of $L$ are naturally embedded into each other by tensoring with the identity. The limiting algebra $\caA_\infty$ is the norm completion of $\cup_L\caA_L$. The states $\langle \Omega_L,(\cdot)\Omega_L\rangle$ extend to states on $\caA_\infty$.
\begin{cor}\label{Corrolary: TD limit}
Assume that the sequence of states $\langle \Omega_L,(\cdot)\Omega_L\rangle$ on $\caA_\infty$ is weak-* convergent and that the sequence $T_{(-,L)}$ is convergent in $\caA_\infty$, then the limit
\begin{equation*}
\mathrm{Ind}(U,\Omega):= \lim_{L\to\infty}\langle \Omega_L,T_{(-,L)}\Omega_L\rangle
\end{equation*}
exists and is an integer.
\end{cor}
\begin{proof}
The existence of the limit is by assumption. The fact that its value is integer is an immediate corollary of Theorem~\ref{thm:main}.
\end{proof}
Combining this corollary with Proposition~\ref{prop:stability}, we conclude that $\mathrm{Ind}(U,\Omega)$ is stable (namely constant) under continuous deformations not violating the assumptions.

As we already pointed out, the convergence of $T_{(-,L)}$ in $\caA_\infty$ must be understood as a locality property: if the condition of the above corollary holds, there exists a limiting operator $T_{-}$ which can be approximated in norm by the strictly local $T_{(-,L)}$. This always holds in one dimension by Assumption~(ii).  In the two-dimensional setting of the quantum Hall effect the index is the Hall curvature, which in turn is a local observable, see also Theorem~1.4 and its proof in~\cite{OurQHE}.
There are however situations where $T_-$ cannot be expected to converge, the typical one being the Lieb-Schultz-Mattis setting in dimensions larger or equal to two: in this case, the index is the total charge in a slab of volume $L^{d-1}$, which should of course not converge to a finite number.

\subsubsection{The assumption of `local charge fluctuations'}

Assumption~(iv) is clearly the most unfamiliar one. Finding a general condition for its validity remains an open question. The notable exception is that of states $\Omega$ that are ground states of a charge conserving Hamiltonian $H$ which satisfies the following condition: For any local $O$ with $\langle\Omega, O\Omega\rangle=0$ there is a  quasilocal $\widetilde{O}$  such that $O\Omega = [H,\widetilde{O}]\Omega$. This condition, which is proven for gapped $H$~\cite{Hastings:2004go}, ensures that one can construct the operators $K_-,K_+$, see Section~\ref{sec: case of gapped} below. The idea is that the operator $\overline Q := Q - K_- - K_+$ gives rise to the same mean charge transport as $Q$, namely $\langle \Omega, T_{-}\Omega \rangle$, because the expectation value of $U^*K_{-}U-K_{-}$ vanishes by invariance of the state under $U$. 
This is analogous to the fact that changing $k$ in \eqref{FermionsConv} by local, i.e.\ trace-class perturbations, does not change the value of the index. 

\subsubsection{The Laughlin argument revisited}\label{sec: flux threading}
Our proof of Theorem~\ref{thm:main} utilizes a version of the Laughlin's flux threading argument. Using (\ref{A:QVar}) and (\ref{Decomp of Q}) we can express the action of a gauge transformation $\ep{\iu \phi Q}$ on the ground state as
\begin{align*}
\ep{\iu \phi Q} \Omega &\sim \ep{\iu \phi Q} \ep{-\iu\phi (Q-K_--K_+)}\Omega + \caO(L^{-\infty}) \\
&=\left(\ep{\iu \phi Q_-} \ep{- \iu \phi (Q_- - K_-)}\right) \left(\ep{\iu \phi Q_+} \ep{- \iu \phi (Q_+ - K_+)}\right) \Omega + \caO(L^{-\infty}) \\ &=: F_-(\phi) F_+(\phi) \Omega + \caO(L^{-\infty}).
\end{align*}
where $\sim$ denotes an equality up to a phase. The operations $F_\pm(\phi)$ correspond to the adiabatic change of a gauge potential by $\pm \phi$ across the line $\partial_{\pm}$. When acting on the ground state of a gapped Hamiltonian $H$, $F_-(\phi)$ implements the (quasi-)adiabatic flow corresponding to threading the flux $\phi$ across the line $\partial_-$, see \cite{OurQHE} or~\cite{HastingsMichalakis}.  Developing the interpretation further, the operation of inserting the fluxes at both ends and compensating for it by a gauge transformation,
\begin{equation}
\label{2.2.2.1}
\ep{-\iu \phi Q} F_-(\phi) F_+(\phi) = \ep{- \iu \phi (Q - K_- - K_+)},
\end{equation}
leaves the ground state invariant up to a phase by (\ref{A:QVar}) as it should. 

Repeating the same reasoning, we further have
\begin{align*}
\ep{\iu \phi Q} \Omega &\sim \ep{\iu \phi Q} U^*  \ep{- \iu \phi (Q - K_- - K_+)} U \Omega + \caO(L^{-\infty})\\
 & = \left(\ep{\iu \phi Q_-} \ep{- \iu \phi (Q_- + T_- - U^* K_- U)}\right) \left(\ep{\iu \phi Q_+} \ep{- \iu \phi (Q_+ + T_+ - U^* K_+ U)}\right) \Omega + \caO(L^{-\infty}) \\ &=: F_-^U(\phi) F^U_+(\phi) \Omega + \caO(L^{-\infty}),
\end{align*}
where $F^U_-(\phi)= U^* F_-(\phi) U + \caO(L^{-\infty}) $. The exact nature of $U$ is irrelevant for the present discussion, but in the case of the Hall effect, it corresponds to the change of gauge potential across a line orthogonal to $\partial_\pm$. Using (\ref{2.2.2.1}) and the corresponding equation for $F^U_\pm$ we find that the operation of inserting flux $\phi$ at both ends, compensating by a gauge transformation, acting by $U$ and reverting this leaves the ground state invariant, see (\ref{TTmTp}). Focusing on the strip $S_-$, this means that the operation gives the ground state back up to a phase factor,
$$
\ep{\iu \phi Q_-} F^U_-(-\phi) \ep{-\iu \phi Q_-} F_-(\phi) \Omega = \chi(\phi) \Omega + \caO(L^{-\infty}).
$$
A computation, see Lemma~\ref{lem:the phase understood}, shows that the argument of $\chi(\phi)$ is linear in $\phi$ with the slope given by the index $\langle\Omega, T_- \Omega\rangle$. On the other hand at flux $\phi=2 \pi$, the compensating gauge transformation is identity, and hence $F_-(2 \pi)$ by itself leaves $\Omega$ invariant up to a phase.  Consequently, at $\phi=2\pi$,  the phase factor  $\chi(2\pi)$ has to be equal to $1$. This establishes integrality of the index. This reasoning is expanded to a sequence of Lemmas in Section~\ref{sec:Proofs}.

This method was previously used with a different gauge transformations and flux threading operators by \cite{oshikawa2000commensurability,lu2017filling} to derive Lieb-Schultz-Mattis theorem and filling constraints in integer and fractional quantum Hall effect. The local flux threading operators that we use originate in~\cite{HastingsMichalakis}.

\subsection{The case of a unique gapped ground state}\label{sec: case of gapped}

We consider the case, relevant for all the examples below, where the state $\Omega$ is the unique gapped ground state of a local charge conserving Hamiltonian. Let the Hamiltonian $H$ be of the form
\begin{equation*}
H = \sum_{X\subset\Lambda}\Phi_X,
\end{equation*}
where $\Phi_X\in\caA_X$ and satisfying
\begin{enumerate}
\item \emph{finite range condition}: There is $R<\infty$ such that $\Phi_X = 0$ if $\mathrm{diam}(X)>R$, 
\item  \emph{finite interaction strength}: There is $m<\infty$ such that $\norm{\Phi_X} \leq m$ for all $X$.
\end{enumerate}
Note that in the case of fermions, the `interaction' $\Phi$ also contains the hopping terms. The dynamics generated by $H$ satisfies a Lieb-Robinson bound. In particular, for any fixed $t$, the Schr\"odinger propagator $\ep{-\iu t H}$ is an almost local unitary in the sense of~(\ref{DefLocalUnitary}).
The assumption of charge conservation means that
\begin{equation*}
X\subset Y\quad\Longrightarrow\quad[\Phi_X,Q_Y] = 0,
\end{equation*}
from which we conclude that $[H,Q_Y]$ is supported in $(\partial Y)^R$, where $\partial Y=\{x\in\Lambda : \dist(x,Y) \leq 1 \text{ and } \dist(x,Y^c) \leq 1  \}$.

By adding a suitable constant to each nonzero $\Phi_X$, we assume that
\begin{equation*}
H\Omega = 0.
\end{equation*}
Let $g>0$ be a spectral gap of $H$, namely
\begin{equation*}
(0,g)\cap \mathrm{Spec}(H) = \emptyset.
\end{equation*}
\begin{prop}\label{prop:Ks}
Let $0$ be a non-degenerate eigenvalue of $H$, and let 
$\Omega$ be the corresponding eigenvector. Assume that the constants $R$, $m$ and $g$ can be chosen independent of $L$. Then Assumptions~(iv,v)  (before Theorem \ref{thm:main}) hold.
\end{prop}
In other words, for a gapped ground state, Theorem \ref{thm:main} yields an index for any almost local charge conserving $U$ for which the ground state is an approximate eigenvector. 
Note that, just as in the previous section, the uniformity of bounds is the only connection between the objects for different $L$.

\begin{proof}
Assumption~(v), i.e.\ exponential clustering, was proved in~\cite{Hastings:2004go}, see also~\cite{BrunoLSM}. So we turn to the construction of $\K_\pm$. Charge conservation implies that $[Q,H]$ is a sum of two terms supported in the strips  $\partial^R_\pm$. We denote
\begin{equation*}
[Q,H] = J_- + J_+.
\end{equation*}
Using the map $\caI$ defined in Section~4 of~\cite{OurMathAdiabatic} (or~\cite{TeufelAd} in the case of fermions), we define
\begin{equation}\label{Ks}
\tilde K_\pm := \iu\caI(J_\pm).
\end{equation}
If $P$ is the projector on $\Omega$,
\begin{equation*}
[Q,P] - \iu [\caI([Q,H]),P] = 0,
\end{equation*}
see Proposition~4.1 of~\cite{OurMathAdiabatic}, and hence $\Omega$ is an exact eigenvector of $Q -  \tilde K_+ - \tilde K_-$. The almost locality of the map $\caI$ implies that one can find approximants $K_\pm \in \caA_{\partial_\pm^{L/16}} $ satisfying
$ \tilde K_+ = K_\pm + \caO(L^{-\infty})$. The boundedness of the map $\caI$ yields $\norm{K_{\pm}} = \caO(L^d)$.  Hence Assumption~(iv) of Theorem~\ref{thm:main} is indeed satisfied.
\end{proof}

\subsection{Unitary $U$ connected to identity.} \label{sec: trivial unitaries}

In most of our examples the unitary $U$ is itself topologically connected to the identity, leading to a slight simplification. 
The natural framework to discuss this is to assume that there is a family $\{G(s): s\in [0,1]\}$ of local Hamiltonians in the sense of Section \ref{sec: case of gapped}: $G(s)=\sum_{X\subset\Lambda} \Psi_X(s)$ with corresponding range $R$ and interaction strength $m$ that are independent of both $s$ and $L$.
We also assume local charge conservation in the sense of Section \ref{sec: case of gapped}.
The unitary $U$ is then given as the time-$1$ solution $U(1)$ of 
\begin{equation} \label{eq: time dependent g}
\iu\partial_s U(s) = G(s) U(s),\qquad  U(0)=1.
\end{equation}
Let now $U_{\pm}(s)$ be the solution to \eqref{eq: time dependent g} with $G(s)$ replaced by $G_{\pm}(s)=\sum_{X \subset \partial_{\pm}^{L/12}} \Phi_X(s)$, and let 
\begin{equation}\label{T-connectedto1}
T_{\pm} =   U^*_{\pm}   Q U_{\pm}-Q,\qquad U_{\pm}:= U_{\pm}(1).
\end{equation}
Since $U_{\pm}$ are by construction supported on $\partial_{\pm}^{L/12}$ and recalling the definition~(\ref{Decomp of Q}), this is also equal to $U^*_{\pm}   Q_\pm U_{\pm}-Q_\pm$, making the strict locality of $T_\pm$ clear.
\begin{prop}\label{prop:Nu if connected to 1}
The operators $T_{\pm}$ defined above indeed satisfy
$$
U^* QU-Q= T_++T_- + \caO(L^{-\infty}),\qquad T_{\pm} \in \caA_{S_{\pm}}.
$$
\end{prop}
\begin{proof}
We write
\begin{align*}
U^* QU-Q &= -\iu\int_0^1  U^* (s) [G(s), Q] U(s) \dd s  \\
& =  -\iu\int_0^1  U^* (s) [G_{-}(s) , Q] U(s) \dd s 
 -\iu\int_0^1  U^* (s) [G_{+}(s) , Q] U(s)  \dd s
\end{align*}
where the second equality is by local charge conservation and the finite range condition. The commutators are supported in $\partial_{\pm}^{R}$ respectively. We now invoke the Lieb-Robinson bound for the dynamics implemented by $U(s)$ combined with Duhamel's principle to replace $U(s)$ by $U_{\pm}(s)$, up to $\caO(L^{-\infty})$. The resulting expression is $T_{+}+T_-+ \caO(L^{-\infty}) $, as required. The constraints \eqref{eq: choice t} are satisfied. Indeed, the spectrum of  $Q_-+T_-= U^*_{-}   Q_- U_{-}$ equals, by unitary invariance, $\mathrm{Spec} (Q_-) \subset \bbZ$.
\end{proof}

Incidentally, we need a slight generalization of the above proposition, allowing for terms of arbitrary support in $G(s)$ but such that the Lieb-Robinson bound continues to hold.
\begin{cor}\label{cor: u connected}
Let the family $G(s)=\sum_{X\subset\Lambda} \Phi_X(s)$ be such that for each $R<\infty$,  the truncated  $G^{(R)}(s)=\sum_{X\subset\Lambda, \mathrm{diam}(X) <R} \Phi_X(s)$ satisfies all assumptions above and 
$$
\frac{1}{|\Lambda|}\sum_{\mathrm{diam}(X) \geq R} \norm{ \Phi_X(s)} = \caO(R^{-\infty}).
$$
Then the conclusion of Proposition \ref{prop:Nu if connected to 1} holds without change.
\end{cor}
We omit the straightforward proof.  


\section{Examples}\label{sec:Examples}



Before moving on to the proof of the main theorem, we now illustrate its applications sketched in the introduction, in details. We will not repeat the setting and refer implicitly to the concepts and notations of Section~\ref{sec:Index}. In each case, it suffices to provide a unitary and a state and to check the invariance of the latter under the action of the former. If the state is not a gapped ground state, then locality of charge fluctuations must be proved independently, too. In all cases, the theorem then provides an integer-valued index. Although it always expresses the charge transported across a hyperplane of codimension $1$, the physical interpretation of the unitary depends on the system. We show that in a quantum Hall setting, the index is equal to the adiabatic curvature, and hence to the Hall conductance.

Although the index is well-defined for any choice of decomposition $U\str Q U - Q = T_- + T_+$ satisfying the constraint \eqref{eq: choice t}, there is usually a natural choice for it. We shall exhibit $T_-$ in each example.
\subsection{Example 1: Index of projections}\label{sec: example one}
In this example, we strive for the simplest setup and we assume $d=1$. For a system of non-interacting fermions with one-particle Hamiltonian $h$ acting on $l^2(\Lambda)$, the state is the gauge-invariant quasi-free state corresponding to Fermi projection $p=1(h\leq \mu)$.
 If the chemical potential $\mu$ lies in a spectral gap, then the state is clustering. The same holds in a disordered setting, provided $\mu$ is in a mobility gap, see~\cite{AizenmanGraf}.
The one-particle $u$ is any unitary that commutes with the Fermi projection. Its locality is expressed in terms of superpolynomial decay of the matrix elements $\vert u_{x,y}\vert$ in $\vert x - y\vert$. If we assume that all fermions carry unit charge, then
 $Q$ is the number of fermions on the half torus $\Gamma$, i.e.\ the second quantization of the indicator function $q=1_{\Gamma}$. It then follows from the decay of $\vert u_{x,y}\vert$ that the matrix elements of $t= u\str q u - q$ have superpolynomial decay in the distance to the boundary of $\Gamma$, which is the one-particle version of local charge conservation.
 
We compute 
$$
UQU^*-Q=
\Gamma(u)\str\dd\Gamma(q)\Gamma(u) - \dd\Gamma(q) = \dd\Gamma(t),\qquad t= u\str q u - q.
$$ 
In contrast to the general many-body setting, there is a canonical way of restricting operators to a  region $Z$ namely, by projecting $\caL(\ell^2(\Gamma))\to \caL(\ell^2(X)): m \mapsto m_{X}=1_Xm 1_X$. Hence we have the splitting $t=t_{S_{-}}+t_{S_{+}}+\error$ and so we set  
\begin{equation*}
T_- = \dd\Gamma(t_{S_{-}})= \dd \Gamma((u\str q u - q)_{S_-}).
\end{equation*}
Let us now determine the spectrum of $\dd\Gamma(q_{S_-}+t_{S_{-}})=\dd\Gamma((u^*qu)_{S_-})$.
We write
$$
u^*qu = (u^*qu)_{S_-}+q_{S_m} +(u^*qu)_{S_+}+\error.
$$ 
Therefore, since both $u^*qu$ and $q_m$ have spectrum $\{0,1\}$, the operators $(u^*qu)_{S_{\pm}}$ have spectrum in $\{0,1\}$ as well. In second quantization, this means that $\mathrm{Spec}(Q_{\pm}+T_{\pm}) \subset \bbN$ and hence \eqref{eq: choice t} is satisfied.    Theorem~\ref{thm:main} reads hence
$$
\langle\Omega,T_-\Omega\rangle=   \tr(p (u\str q u - q)_{S_-})   \in \fatlattice.
$$

Let us now connect this to the infinite-volume setting and to the `index of projections', as introduced in \cite{ASSIndex}. The finite-volume version of the index of projections is discussed in Appendix~C of~\cite{KitaevHoneycomb}. There, the index is called the `flow of a unitary' and the concrete example of the shift by $m\in\bbZ$ on $l^2(\bbZ)$ shows that it can take all possible integer values.

\subsubsection{Infinite-volume setting}
We now consider from the one-particle Hilbert space $\ell^2(\bbZ)$ and we imagine the finite-volume operators $h_L,u_L$ discussed above (where we did not indicate the $L$-dependence explicitly) to act on $\ell^2(\bbZ)$ by the natural embedding $\ell^2(\Lambda)\to \ell^2(\bbZ)$.
As already explained, our setup does not assume nor imply any relation between different volumes $L$, except for the uniformity of some bounds. However, in this section only, we do want to consider an infinite volume setup and so we assume that $h_L,u_L$ converge to infinite-volume operators $h,u \in \caL(\ell^2(\bbZ))$ in the sense that
$$
h_L\phi\to h\phi,\qquad u_L\phi\to u\phi,
$$
for any compactly supported $\phi \in \ell^2(\bbZ)$. We moreover assume (although some of this actually follows from the construction above) the following:
\begin{enumerate}
\item $h$ has range $R<\infty$, i.e.\ it is a sum of terms acting each on at most $R$ consecutive lattice sites. The terms are uniformy bounded so $\norm{h}\leq C$
\item The unitary $u$ is almost local in the sense that $u_{x,y}=\caO(|x-y|^{-\infty})$. 
\item  The chemical potential $\mu$  lies in a spectral gap of $h$.
\item  The Fermi projection $p=1(h\leq \mu)$ commutes with $u$: $[p,u]=0$. 
\end{enumerate}
Finally, we choose $q=1_\bbN$, corresponding to the picture that we let the boundary of $\Gamma$ at $L/2$ disappear at infinity.  By the locality of $u$, the operator $u^*qu-q$ is compact, in particular it is trace-class, and  it follows that the finite-volume expression $\tr (p_L({u_L}^*1_\Gamma u_L-1_\Gamma)_{S-})$ converges to the infinite-volume expression
$
\tr (p(u^*qu-q)).
$
Note that the disappearance of one boundary comes at the price of an additional, necessary, subtlety: the operators $u^*qu$ and $q$ are individually not trace-class, even though their difference is. 
Hence, the upshot is that our result implies that 
\begin{equation*}
\tr(p (u\str q u - q)) \in \bbZ.
\end{equation*}
As $[p,u]= 0$, this is the expression proposed in~(\ref{FermionsConv}). Since $p,q$ are not commuting, the operator $pqp$ is not a projection. Therefore, it remains to prove the claim that $pqp$ can be deformed to a projection without changing the value of the index. 
\begin{lemma}
There exists a $\chi\in C^\infty(\bbR)$ such that $\chi(x) = 0$ for $x\leq 0$, $\chi(x) = 1$ for $x\geq 1$ and such that $\chi (p q_\V p) - p q_\V p \in\caI_1$. In particular,
\begin{equation*}
\tr(u\str p q_\V p u - p q_\V p) = \mathrm{Index}( \chi(p q_\V p); u) \in\bbZ,
\end{equation*}
whenever the operator on the left-hand side is in $\caI_1$.
\end{lemma}
\begin{proof}
Since $\mu$ lies in the gap and $\norm{h} < \infty$, the spectral projection $p$ can be obtained from a compactly supported $C^\infty$-bump function~$\theta$, namely
\begin{equation*}
p = \theta(h).
\end{equation*}
We decompose $h$ as
\begin{equation*}
h = h_\mathrm{d} + h_\mathrm{o}, \qquad h_\mathrm{d}= qhq+(1-q)h(1-q),
\end{equation*}
and note that $h_\mathrm{o}$ is a finite-rank operator composed of finite number of interaction terms connecting $x<0$ to $x\geq 0$. We now claim that
\begin{equation} \label{eq: decomp of chi statement}
\theta(h) - \theta(h_\mathrm{d}) \in\caI_1.
\end{equation}
Indeed, the Fourier transform $\hat \theta(t)$ is smooth and $\caO(|t|^{-\infty})$. By the spectral theorem,
\begin{equation*}
\theta(h) - \theta(h_\mathrm{d})
= \frac{\iu}{\sqrt{2\pi}}\int_{-\infty}^{\infty} \hat \theta(t) \int_0^t \ep{\iu (t-s) h}h_\mathrm{o}\ep{\iu s h_\mathrm{d}}\dd s \dd t.
\end{equation*}
The claim~(\ref{eq: decomp of chi statement}) follows since $h_\mathrm{o}$ has finite rank and the integrals are convergent in the trace norm. Therefore,
\begin{equation}\label{compact}
pqp - \theta(h_\mathrm{d}) q \theta(h_\mathrm{d}) \in\caI_1.
\end{equation}
Since $\theta(h) = p$ is a projection, it follows by Weyl's theorem on compact perturbations that the spectrum of $\theta(h_\mathrm{d})$ has a gap in $[0,1]$. Moreover, $[\theta(h_\mathrm{d}),q] = 0$ by construction, so that $\theta(h_\mathrm{d}) q \theta(h_\mathrm{d})$ also has a gap. Using~(\ref{compact}) and Weyl's theorem again, $pqp$ as well has a gap, say the interval $I\subset [0,1]$. It remains to pick a smooth function $\chi$ interpolating from $0$ to $1$ within~$I$ to conclude that
\begin{enumerate}
\item $\chi(pqp)$ is a projection, and that
\item  $\chi(pqp)-pqp \in\caI_1$ by the same reasoning as that leading to~(\ref{eq: decomp of chi statement}). 
\end{enumerate}
We conclude by noting that perturbations of $pqp$ that are trace-class do not contribute to $\tr(u\str p q_\V p u - p q_\V p)$.
\end{proof}
%

\subsection{Example 2: The Lieb-Schulz-Mattis theorem} \label{sec: example two}
We now turn to interacting quantum spin systems and explain how the Lieb-Schulz-Mattis (LSM) theorem, in any dimension, is a corollary of our index theorem.
As already said, $U$ implements the translation $x_1\mapsto x_1+1$ and we assume that the translation-invariant Hamiltonian has a non-degenerate gapped ground state $\Omega=U\Omega$. Let us write $[a,b]=\{a \leq  x_1 \leq b \}$ and $[a]=[a,a]$, so that $\Gamma=[0,L/2-1]$. Then 
$$ U\str Q_\V U - Q_\V = -Q_{[0]}+Q_{[L/2]}, $$
and the natural definition of $T_\pm$ is 
$$
T_-= -Q_{[0]},\qquad T_+=  Q_{[L/2]}.
$$
With this choice, $Q_-+T_-= Q_{\Gamma\cap(S_- \setminus [0])}$,  whose spectrum is obviously integer, hence \eqref{eq: choice t} is satisfied.
Theorem \ref{thm:main} now yields 
$$
\langle \Omega, Q_{[0]} \Omega \rangle \in  \bbZ_{\scriptscriptstyle{(\caO(L^{-\infty}))}}
$$  
which is the desired result. 

For a spin chain, this is immediately the quantization of the filling factor. The higher dimensional case is slightly more involved as it requires a non-commensurability condition on the number of sites in each direction, see the discussion in~\cite{oshikawa2000commensurability}. Similarly, the `original' LSM Theorem on the non-existence of unique gapped ground states for half-odd-integer spin systems is easily recovered by considering as the charge a $\mathrm{U}(1)$-subgroup of $\mathrm{SU}(2)$, and a system having an odd number of sites. This connection is due to \cite{HastingsLesHouches}.

\subsection{Example 4: Thouless pumps} \label{sec: example four}
By a Thouless pump, we mean a system adiabatically undergoing a cyclic change. The smooth family of local Hamiltonians along the cycle is denoted $H(s)$ with $H(0) = H(1)$. If $H(s)$ has a unique gapped ground state $\Omega(s)$ for all $s\in[0,1]$, then the quasi-adiabatic generator
\begin{equation*}
G(s) := \caI(\dot H(s)) = \int_{-\infty}^{+\infty} W(t) \ep{\iu t H(s)}  \dot H(s) \ep{-\iu t H(s)}\dd t,\qquad \dot H(s) = \frac{\dd}{\dd s}H(s).
\end{equation*}
introduced in~\cite{Hastings:2004go} yields a family of unitaries $U(s)$ such that $\Omega(s) = U(s)\Omega(0)$, provided $W$ satisfies certain properties~\cite{HastingsWen,Sven}. In particular $\Omega(0) = \Omega(1) = U(1)\Omega(0)$, so the index theorem applies to $U = U(1)$ and $\Omega = \Omega(0)$. The locality of $U$, Assumption~(i), follows from the Lieb-Robinson bound and decay properties of $W$. If every $H(s)$ is a local Hamiltonian conserving charge, see Section~\ref{sec: case of gapped}, then $G(s)$ satisfies the assumptions of Corollary~\ref{cor: u connected}, which provides explicit $T_\pm$ satisfying the constraint~\eqref{eq: choice t}. It is this application which requires the corollary on top of Proposition~\ref{prop:Nu if connected to 1}.

The interpretation of our index theorem in this case is now clear: $\langle \Omega, T_- \Omega\rangle$ is the average charge transported across the line $\partial_-$ in one adiabatic cycle, and it is quantized $\langle \Omega, T_- \Omega\rangle \in \fatlattice$, as was first pointed out in a noninteracting setting in~\cite{Thouless83}.

Incidentally, the unitary considered here remains a useful tool even in other cases, for example when the spectral gap may close along the cycle as demonstrated e.g.\ in~\cite{HastingsLSM,HastingsMichalakis}.

\subsection{Example 3: Hall conductance} \label{sec: example three}

Here we take $d=2$, a unique gapped ground state $\Omega$ and the unitary $U$ corresponds to threading of a unit of flux along $x_1$. To describe this, we introduce some additional notation. We recall that notions like $\Gamma, S_{\pm}, K_{\pm},\ldots$ were defined based on restrictions on the coordinate $x_1$.
We now define the analogous notions based on the coordinate $x_2$ and we make this explicit by  endowing these symbols with the superscripts $(1), (2)$. For example
$$\Gamma^{(i)} = \{ 0 \leq x_i < \tfrac{L}{2}\},\qquad  S^{(i)}_-=\{ -\tfrac{L}{8} \leq x_i \leq \tfrac{L}{8} \}, \qquad
Q^{(i)}= Q_{\Gamma^{(i)}}, $$
for $i=1,2$. Then, we take $U$ to be 
$$
U = \ep{- 2 \pi \iu (Q^{(2)} - K^{(2)}_- )}= \ep{- 2 \pi \iu (Q^{(2)}_- - K^{(2)}_- )},
$$
where $K^{(i)}_\pm$ are provided by Proposition~\ref{prop:Ks}. The second equality is because of $Q^{(2)}_m, Q^{(2)}_+$, defined as in~(\ref{Decomp of Q}), have integer spectrum and commute with $Q^{(2)}_- - K^{(2)}_-$.

We note that $U$ is connected to the identity in the sense of Section~\ref{sec: trivial unitaries} with a time independent generator 
\begin{equation}\label{Hall G}
G = (Q^{(2)}_- - K^{(2)}_-).
\end{equation}
Hence there is again a natural choice for $T_-$ satisfying \eqref{eq: choice t}. We further find that $\var_\Omega(U)=\caO(L^{-\infty})$: this is physically clear since $U$  implements the threading of a unit of flux (see \ref{sec: flux threading}, also \cite{OurQHE} or~\cite{HastingsMichalakis}), while we shall see in Section~\ref{sec:Proofs} that it is a mathematical fact that follows from the clustering of $\Omega$ and the property of local charge fluctuations. 

All assumptions of Theorem \ref{thm:main} are satisfied. In fact, technically speaking, this setup is an instance of a Thouless pump described in Section~\ref{sec: example four}. As discussed there, the charge transported across the fiducial line is an integer. It remains to show that the index equals $2\pi$ times the Hall conductance to conclude that the Hall conductance is quantized in this many-body setting. We do this here starting from the following expression for the Hall conductance:
\begin{equation}\label{curvature}
 \kappa := \iu \langle \Omega, [K^{(1)}_-, K^{(2)}_-] \Omega \rangle,
\end{equation}
which is derived from the well-known formula $\iu\Tr(P\mathrm dP\wedge \mathrm dP)$ for the adiabatic curvature, see~\cite{AvronSeiler85,OurQHE}.
\begin{thm}\label{thm: hall}
Let $\kappa$ and $\Omega, T_-$ be defined above. Then
\begin{equation*}
2 \pi \kappa - \langle\Omega, T_- \Omega \rangle = \caO(L^{-\infty}).
\end{equation*}
In particular $2\pi\kappa \in \fatlattice$.
\end{thm} 

\begin{proof}

We recall from Section~\ref{sec: trivial unitaries} that
\begin{equation*}
T_- = e^{2 \pi \iu G_-} Q^{(1)} e^{-2 \pi \iu G_-}  - Q^{(1)},
\end{equation*}
where $G_-$ is a restriction of $G$ to the region $\partial_{-}^{{L}/{12}}= \{|x_1| \leq \tfrac{L}{12}\}$, see Figure~\ref{fig:QHE sets}. We shall repeatedly use that (a) $G_-$ is supported in the corner $S^{(2)}_- \cap \partial_{-}^{L/12}$, and (b) it is the restriction of an operator that can be written as sums of local terms, see~(\ref{Ks}).
\begin{figure}
\centering
\includegraphics[width = 0.4\textwidth]{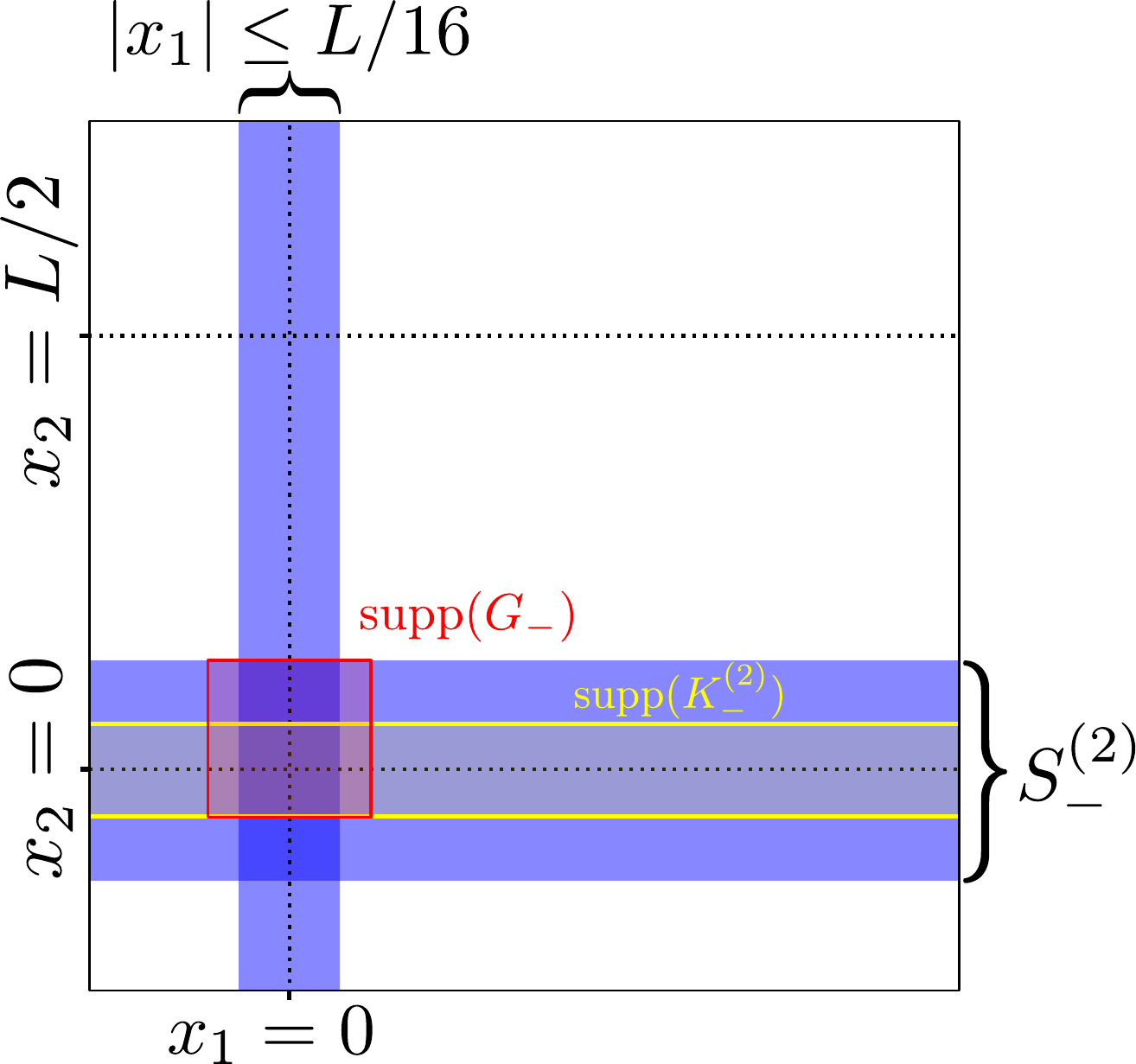}
\caption{The strips used in the quantum Hall setting.}\label{fig:QHE sets}
\end{figure}

Now, the fundamental theorem of calculus yields
\begin{align}
\langle\Omega,T_-\Omega \rangle  &= \iu \int_0^{2 \pi} \langle\Omega,   \ep{\iu \phi  G_-} [G_-, Q^{(1)} ] \ep{-\iu \phi G_-  }\Omega  \rangle \mathrm{d} \phi \nonumber\\
&=
\iu \int_0^{2 \pi} \langle\Omega,   \ep{\iu \phi  G} [G_-, Q^{(1)} ] \ep{-\iu \phi G }\Omega  \rangle \mathrm{d} \phi + \caO(L^{-\infty}),\label{Some QHE Eq}
\end{align}
where the second equality follows from the Lieb-Robinson bound as in the proof of Proposition~\ref{prop:Nu if connected to 1}. Now we use the specific form of $G$, in particular the identity
\begin{equation*}
\ep{\iu\phi(Q^{(2)} - K^{(2)}_- - K^{(2)}_+)}
=\ep{\iu\phi(Q^{(2)}_- - K^{(2)}_-)}\ep{\iu\phi Q^{(2)}_m}\ep{\iu\phi(Q^{(2)}_+ - K^{(2)}_+)}.
\end{equation*}
Since the last two factors on the right hand side commute with any $A$ supported in $S^{(2)}_{-}$ we conclude that for any such $A$,
$$
\langle \Omega, \ep{\iu\phi (Q^{(2)}_- - K^{(2)}_-)} A \ep{ -\iu\phi(Q^{(2)}_- - K^{(2)}_-)} \Omega \rangle = \langle \Omega, A \Omega \rangle + \caO(L^{-\infty})
$$
by the invariance~(\ref{A:QVar}) of $\Omega$ under $Q^{(2)} - K^{(2)}_- - K^{(2)}_+$. Using this with $A = G_-$ in~(\ref{Some QHE Eq}), see Property~(a) above, we obtain
$$
\langle\Omega,T_-\Omega \rangle = 2 \pi \iu \langle \Omega, [G_-, Q^{(1)}] \Omega \rangle + \caO(L^{-\infty}).
$$
Using~(\ref{A:QVar}) now for $Q^{(1)} - K^{(1)}_- - K^{(1)}_+$ and by the fact that $K^{(1)}_+$ commutes with $G_-$, we get
\begin{align*}
\langle \Omega, [G_-, Q^{(1)}] \Omega \rangle 
&= \langle \Omega, [G_-, K^{(1)}_-] \Omega \rangle  +  \caO(L^{-\infty}) 
= \langle \Omega, [G, K^{(1)}_-] \Omega \rangle +  \caO(L^{-\infty}) 
\end{align*}
where we used Property (b) in the second equation. The restriction $Q^{(2)}_{--}$ of $Q^{(2)}_{-}$ to the corner $S^{(2)}_- \cap   \partial_{-}^{{L}/{12}} $ commutes with $K^{(1)}_+$, and so we proceed with the similar reasoning backwards to get
\begin{equation*}
 \langle \Omega, [Q^{(2)}_-, K^{(1)}_-] \Omega \rangle =  \langle \Omega, [Q^{(2)}_{--}, K^{(1)}_-] \Omega \rangle  = \langle \Omega, [Q^{(2)}_{--}, Q^{(1)}] \Omega \rangle +  \caO(L^{-\infty}) = \caO(L^{-\infty})
\end{equation*}
and hence $\langle \Omega, [G, K^{(1)}_-] \Omega \rangle = - \langle \Omega, [K^{(2)}_-, K^{(1)}_-] \Omega \rangle + \caO(L^{-\infty})$. Therefore,
\begin{equation*}
\langle\Omega,T_-\Omega \rangle 
= 2\pi \iu \langle\Omega,  [K^{(1)}_{-}, K^{(2)}_- ] \Omega  \rangle + \caO(L^{-\infty})
\end{equation*}
which is what we had set out to prove, see~(\ref{curvature}).
\end{proof}

For completeness of the discussion, we point out that the present argument proves quantization of the Hall conductance, but it does not explain the experimental observation of plateaux in the graph of the conductance versus the filling factor. We refer the reader to~\cite{AizenmanGraf} for a detailed discussion of the role of disorder in the integer quantum Hall effect.

\subsection{Example 5: Bloch's Theorem}

We now consider a local charge conserving Hamiltonian in the sense of Section~\ref{sec: case of gapped}. Then $[H,Q]$ is a sum of two terms strictly supported on strips of width $R$ around $\partial_\pm$. We denote
\begin{equation*}
\iu [H,Q] = J_- + J_+,
\end{equation*}
the current operators across lines $\partial_\pm$.

We assume that $H$ has a unique gapped ground state and hence local charge fluctuations. Then for any $t$, $\Omega$ is an exact eigenvector of the propagator $U(t) = \ep{-\iu t H}$. For any fixed $t$, $U(t)$ is almost local by the Lieb-Robinson bound, and connected to the identity. The index reads
\begin{equation}\label{BlochBeginning}
\langle\Omega, (U_-(t)\str Q U_-(t) - Q)\Omega\rangle \in\fatlattice,
\end{equation}
see~(\ref{T-connectedto1}). By the argument used in the proof of Proposition~\ref{prop:Nu if connected to 1}, 
\begin{equation*}
(U_-(t) Q U_-(t) - Q) = \int_0^t \langle\Omega, U(s) J_- U(s) \Omega\rangle \dd s + \caO(L^{-\infty}),
\end{equation*}
so that (\ref{BlochBeginning}) is equivalent to
\begin{equation*}
t \langle\Omega, J_- \Omega\rangle \in\fatlattice.
\end{equation*}
Since this vanishes at $t=0$, we have proved the following version of Bloch's theorem.
\begin{thm}
Let $H$ be a gapped, local charge conserving Hamiltonian with a unique ground state $\Omega$. Then
$$
\langle \Omega, J_- \Omega \rangle = \caO(L^{-\infty}).
$$
\end{thm}

We note that, unlike in~\cite{Bohm}, this version of Bloch's theorem does not require time-reversal invariance. However, clustering is an essential ingredient. In the preprint~\cite{watanabeBloch} which appeared shortly after the publication of the present work, this is not needed. The result there is however slightly weaker in two aspects: In higher dimensions it is concerned only with current densities; It only shows that the current has a power law decay in the system size.


\section{Proof of Theorem~\ref{thm:main}}\label{sec:Proofs}

\subsection{Approximate eigenvectors}

The variance of an operator $A$ in a state $\psi$, as defined in \eqref{def: var}, gives a handy way of expressing that $\psi$ is an approximate eigenvector of $A$. We state three lemmata that will be used later on. 

First of all, if $A$ has a small variance in the state $\psi$, then by definition $A$ acts on $\psi$ by multiplication by $\langle \psi, A\psi\rangle$, up to an error that is small in norm. This extends to products of operators having small variance. Specifically,
\begin{lemma}\label{lma:variance of product}
Let $(A_i)_{ i=1,\ldots, k}$ satisfy $\var_\psi(A_i)=\caO(L^{-\infty})$ and $\norm{A_i} =\caO(L^d)$ with $k$ fixed, i.e.\ independent of $L$.  Then
\begin{equation*}
\langle \psi, A_1\cdots A_k \psi \rangle=  \prod_{i=1}^k  \langle \psi, A_i \psi\rangle + \caO(L^{-\infty}).
\end{equation*}
\end{lemma}
\begin{proof}
Let $a_i=\langle \psi, A_i \psi\rangle$.
We note that $ \norm{ (A_i-a_i) \psi}=\caO(L^{-\infty})$ by  $\var_\psi(A_i)=\caO(L^{-\infty})$. We split now $A_i=a_i+(A_i-a_i)$ and apply the bound on $(A_i-a_i) \psi$ recursively. 
\end{proof}
If $\Omega$ is clustering, a somewhat converse statement of the above holds. If a unitary has small variance and factorizes in two unitaries with disjoint and well separated supports, then $\Omega$ is an approximate eigenvector of each of the factors.
\begin{lemma}\label{lem:Var W1W2}
\label{varU}
Consider  unitaries $W_1,  W_2$ with $W_{1,2} \in \mathcal{A}_{X_{1,2}}$ and set $W=W_1W_2$. Let $\psi$ be clustering in $x_1$-direction, see~(\ref{A:clustering}). Then
\begin{equation*}
\var_\psi(W_1)  \leq \var_\psi(W) + 2 C \vert X_1\vert \vert X_2\vert \ep{-c d_1(X_1,X_2)}.
\end{equation*}
\end{lemma}
\begin{proof}
We first note that for any unitary
\begin{equation}\label{UnitaryVar}
\var_\psi(W) = 1 - \left\vert\left\langle \psi, W \psi\right\rangle\right\vert^2,
\end{equation}
and of course that $\vert \langle \psi, W \psi\rangle\vert\leq1$. By clustering,
\begin{equation*}
\left\vert\left\langle \psi, W \psi\right\rangle\right\vert 
\leq C\vert X_1\vert \vert X_2\vert \ep{-c d_1(X_1,X_2)}
+\left\vert\left\langle \psi, W_1 \psi\right\rangle\right\vert
\end{equation*}
The square of this inequality yields the claim by~(\ref{UnitaryVar}).
\end{proof}
Finally, we note the following variational characterization of $\var_\psi(A)$:
\begin{equation*}
\var_\psi(A) = \inf\left\{a\in\bbC:\Vert (A - a)\psi\Vert^2\right\}
\end{equation*}
since the infimum is reached at $a = \langle \psi, A\psi\rangle$.

\begin{lemma}\label{lma:varf}
Let $f\in C^1(\bbR;\bbC)$ and let $A$ be a self-adjoint operator. Then
$$
\var_\psi (f(A)) \leq \sup_{x \in \mathbb{R}} |f'(x)|^2 \var_\psi (A).
$$
\end{lemma}
\begin{proof}
For any $\lambda\in\bbR$ we have by functional calculus
$$
(f(A) -f(\lambda)) \psi = \int_0^1 f'(t (A - \lambda) + \lambda) (A - \lambda) \psi \dd t.
$$
By the variational formula above, this gives
$$
\var_\psi (f(A)) \leq \sup_{x \in \mathbb{R}} |f'(x)|^2 \|(A -\lambda) \psi \|^2.
$$
Taking $\lambda = \langle \psi, A \psi \rangle$ finishes the proof.
\end{proof}


\subsection{Local charge fluctuations}

By assumption, the state $\Omega$ has local charge fluctuations, so we can modify $Q$ (and therefore also $ Q^U:=U^*QU$) at the boundaries of its support $\Gamma$ to obtain operators $\overline{Q},{\overline{Q}^U} $ that have $\Omega$ as an approximate eigenvector. This is the content of the next lemma.

We recall the operators $\K_\pm\in\caA_{\partial_\pm^{L/16}}$  given by~(\ref{A:QVar}). We write $K_\pm^U  \in\caA_{S_{\pm}}$ for local approximations of $U^*K_\pm U$ which, by~(\ref{DefLocalUnitary}) satisfy
\begin{equation}\label{KU}
\norm{K_\pm^U- U^*K_\pm U}  = \caO(L^{-\infty}).
\end{equation}
Note that the bound holds even though both the norm and the support of $K_\pm$ grow with $L$ because of the almost exponential locality of $U\str(\cdot) U$ and our choice of strips $S_{\pm}$ with $L$-dependent widths, see~(\ref{DefLocalUnitary}) and Figure~\ref{fig: torus}. Finally, we recall $T_\pm $ introduced in Assumption~(ii).
\begin{lemma}\label{lma:Qbar}
We define
\begin{equation*}
\overline{Q} := Q - K_- -  K_+,
\end{equation*}
\begin{equation}\label{Qbar}
{\overline{Q}^U} := Q +T_-+T_+ -  K_-^U - K_+^U.
\end{equation}
Then 
\begin{equation*}
{\overline{Q}^U} - U\str \overline Q U = \caO(L^{-\infty}), \qquad 
\ep{\iu\phi {\overline{Q}^U} } - U\str \ep{\iu\phi \overline{Q}} U= \caO(L^{-\infty}),
\end{equation*}
for any fixed $\phi\in\bbR$, and the variances of $\overline{Q},\overline{Q}^U,\ep{-\iu \phi \overline{Q}},\ep{-\iu \phi {\overline{Q}^U}}$ in the state $\Omega$ are all $\caO(L^{-\infty})$.
\end{lemma}

\begin{proof}
The first claim is by~(\ref{A:Uloc}) and~(\ref{KU}). It implies the second one by Duhamel's formula. The claim on the variance of $\overline{Q}$ is by assumption, see~(\ref{A:QVar}). For the variance of $\overline{Q}^U$, we further use that $\Omega$ is an approximate eigenvector of $U$, see~(\ref{A:Invariance}), and Lemma~\ref{lma:variance of product}. Finally, Lemma~\ref{lma:varf} yields the variance of the exponentials.
\end{proof}
We now define decompositions of these operators analogous to~(\ref{Decomp of Q}) into three terms whose supports are respectively in $S_-, S_m , S_+$:
\begin{align}
\overline{Q} &= (Q_- - K_- ) + Q_\mathrm{m} + ( Q_+ - K_+)  \, =:   
\overline Q_- + Q_\mathrm{m} + \overline Q_+ , \label{qbarmiddle} \\
\overline{Q}^U &= (Q_- +T_-- K_-^U ) + Q_\mathrm{m} + (Q_+ +T_+- K_+^U )    \, =:   
\overline Q^U_- + Q_\mathrm{m} + \overline Q^U_+ .\label{qbarmiddleU}
\end{align}
Consider the unitary
\begin{equation}\label{def:T}
Z(\phi) := U^* \ep{\iu \phi \overline Q} U \ep{-\iu \phi \overline Q} = \ep{\iu \phi {\overline{Q}^U}} \ep{-\iu \phi \overline Q} + \caO(L^{-\infty}),
\end{equation}
where the second equality is by Lemma~\ref{lma:Qbar}. By Lemma~\ref{lma:variance of product}, the vector $\Omega$ is an approximate eigenvector of $Z(\phi)$, with an eigenvalue in $U(1)$ that is approximatively independent of $\phi$, hence equal to $1+\caO(L^{-\infty})$.  We will now argue that this constant and trivial phase can be decomposed in two non-trivial motions, taking place around the boundaries~$\partial_{\pm}$. Indeed, \eqref{def:T} is a product of two unitaries supported in  $\caA_{S_{\pm}}$,
\begin{equation}\label{TTmTp}
Z(\phi) = Z_-(\phi) Z_+(\phi) + \caO(L^{-\infty})
\end{equation}
with
\begin{equation}\label{Tpm}
Z_\pm(\phi) := \ep{\iu\phi  {\overline Q^U_{\pm}} }\ep{-\iu\phi \overline Q_{\pm}}.
\end{equation}
By Lemma~\ref{lem:Var W1W2}, $\Omega$ is an approximate eigenvector of $Z_\pm(\phi)$ as well, and we call the corresponding eigenvalue
\begin{equation}\label{Defchi}
\chi(\phi) := \left\langle\Omega, Z_-(\phi)\Omega\right\rangle.
\end{equation}
This $\chi(\phi)$ need not be independent of $\phi$, and we now show that it gives rise to the many-body index.

\subsection{A formula for $\chi(\phi)$}
The following lemma connects the phase $\chi(\phi)$ to the quantity $\left\langle \Omega, T_- \Omega \right\rangle$, whose quantization we want to prove. 
\begin{lemma}\label{lem:the phase understood}
For any $0\leq \phi \leq 2\pi$, 
\begin{equation}
\chi(\phi) = \ep{\iu \phi \left\langle \Omega, T_- \Omega \right\rangle} + \caO(L^{-\infty}).
\end{equation}
\end{lemma}
To prove this lemma, we make use of the following tool:
\begin{lemma}\label{lem:tool}
For any $A$ such that $[A,Q_\mathrm{m}] = 0$, we have
\begin{equation}\label{symmetry}
\frac{1}{\Vert A \Vert}\left(\langle\Omega, \ep{\iu\phi  {\overline Q^U_{-}}  } A \ep{-\iu\phi\overline Q_{-}   }\Omega\rangle
- \langle\Omega, \ep{-\iu\phi{\overline Q^U_{+}}} A \ep{\iu\phi{\overline Q_{+}}}\Omega\rangle\right) = \caO(L^{-\infty}).
\end{equation}
\end{lemma}
\begin{proof} Using the decomposition in commuting terms~(\ref{qbarmiddle}) and~(\ref{qbarmiddleU}), we have
\begin{align*}
\ep{-\iu\phi  {\overline Q_{-}}    }\Omega &=   \ep{\iu\phi Q_m}  \ep{\iu\phi {\overline Q_{+}}}   \ep{-\iu\phi\overline{Q}} \Omega, \\
 \ep{\iu\phi  {\overline Q^U_{-}}    }\Omega &=   \ep{-\iu\phi Q_m}  \ep{-\iu\phi {\overline Q^U_{+}}}   \ep{\iu\phi\overline{Q}^U} \Omega.
\end{align*}
Plugging these identities in the left hand side of \eqref{symmetry} and using $\ep{-\iu\phi Q_m} A \ep{\iu\phi Q_m}=A $, we get
\begin{align*}
\langle\Omega, \ep{\iu\phi  {\overline Q^U_{-}}  } A \ep{-\iu\phi\ \overline Q_{-} }\Omega\rangle 
&= \langle   \Omega,  \ep{\iu\phi{\overline{Q}^U}} \ep{-\iu\phi {\overline Q^U_{+}}   } A \ep{\iu\phi{\overline Q_{+}}}  \ep{-\iu\phi\overline{Q}}  \Omega\rangle \\
&= \langle   \Omega,  \ep{-\iu\phi{\overline Q^U_{+}}} A \ep{\iu\phi{\overline Q_{+}}}    \Omega\rangle \, \langle \Omega,   \ep{\iu \phi {\overline{Q}^U}} \ep{-\iu \phi \overline Q} \Omega \rangle  +\caO(L^{-\infty})
\end{align*}
where the equality on the second line follows because $\Omega$ is an approximate eigenvector of both $\ep{-\iu \phi {\overline{Q}^U}}$ and $ \ep{-\iu \phi \overline{Q}}$, see Lemma~\ref{lma:Qbar}, and Lemma~\ref{lma:variance of product}. The second factor is in fact $\langle \Omega,  Z(\phi) \Omega \rangle $, see \eqref{def:T}, and we argued above that this is $1+\caO(L^{-\infty})$.
\end{proof}
We are now ready to give the
\begin{proof}[Proof of Lemma~\ref{lem:the phase understood}]
By applying Lemma \ref{lem:tool} with $A=1$, we get 
\begin{equation}\label{particlehole}
\chi(\phi) = 
\langle\Omega, Z_+(-\phi)\Omega\rangle + \caO(L^{-\infty}).
\end{equation}
We compute
\begin{equation*}
\chi'(\phi) = \mathrm{i}\langle \Omega, \ep{\mathrm{i}\phi  { \overline Q^U_- }} D_-\ep{-\mathrm{i}\phi{ \overline Q_- }}\Omega\rangle
\end{equation*}
where $D_-=  \overline Q^U_-- \overline Q_-=T_-+ K_--K_-^U $. Lemma \ref{lem:tool} with $A= D_-$ now yields 
\begin{align*}
\chi'(\phi) &=\mathrm{i}
\langle \Omega, \ep{-\mathrm{i}\phi{ \overline Q^U_+ }}D_-\ep{\mathrm{i}\phi{ \overline Q_+ }}\Omega\rangle+ \caO(L^{-\infty}) \\
 & =\mathrm{i}\langle \Omega, D_-\Omega\rangle \langle\Omega, Z_+(-\phi)\Omega\rangle + \caO(L^{-\infty})\\
&=\mathrm{i}\langle \Omega, D_-\Omega\rangle \chi(\phi)+ \caO(L^{-\infty}) \\
&=\mathrm{i}\langle \Omega, T_-\Omega\rangle \chi(\phi)+ \caO(L^{-\infty})
\end{align*}
The second equality is by clustering, the third is by \eqref{particlehole}, and the fourth follows by~(\ref{KU}) and $\langle \Omega,K_- \Omega\rangle= \langle\Omega, U K_- U^*\Omega\rangle$, since $\Omega$ is an approximate eigenvector of $U$. 
 Since it is immediate from the definition that $\chi(0) = 1$, Lemma \ref{lem:the phase understood} follows. 
\end{proof}
\subsection{Proof of Theorem \ref{thm:main}} We now finish the proof by computing $\langle \Omega, T_-\Omega\rangle$ from the value of $\chi(2\pi)$.
At $\phi=2\pi$, the constraint $\ep{2\pi \iu  (Q_++T_+) }=1+\vsmall$ and $\ep{2\pi\iu Q_m}=1$ imply that
\begin{equation*}
\ep{2\pi \iu  {\overline Q^U_{-}} } =  \ep{2\pi \iu  (Q_- + T_- - K_-^U)} =  \ep{2\pi \iu  (Q + T_- + T_+ - K_-^U)} +\vsmall,
\end{equation*}
see the definition~(\ref{qbarmiddleU}). By local charge conservation, the exponent is equal to $U\str Q U - K_-^U$, and hence $U\str (Q- K_-) U$, up to $\caO(L^{-\infty})$. It follows that 
\begin{equation*}
 \ep{2\pi \iu  {\overline Q^U_{-}} } 
= U^*\ep{2\pi \iu  \overline Q_-}U+\vsmall,
\end{equation*}
where we used again that $\ep{2\pi\iu Q_m}=1=\ep{2\pi\iu Q_+}$, and hence
\begin{equation}\label{eq: res as product}
Z_-(2\pi)=  U^*\ep{2\pi \iu  \overline Q_-}U \ep{-2\pi \iu  \overline Q_-}+\vsmall.
\end{equation}
Since the factorization $\ep{2 \pi \iu \overline{Q}} = \ep{2\pi\iu \overline{Q}_-} \ep{2\pi\iu\overline{Q}_+}$ holds at $\phi=2\pi$,  Lemma~\ref{varU} implies that $\Omega$ is an approximate eigenstate of $\ep{2\pi\iu  \overline{Q}_-}$.
Therefore, all four operators on the right of \eqref{eq: res as product} have $\Omega$ as approximate eigenvector and it follows by Lemma~\ref{lma:variance of product} that $\chi(2\pi)=\langle Z_-(2\pi)\rangle=1+\vsmall$. 
The theorem follows now directly from Lemma~\ref{lem:the phase understood}.

\bigskip
\bigskip

\noindent \textbf{Acknowledgements.} The authors would like to thank Y.~Ogata for inspiring discussions. This research was supported in part by funding from the Simons Foundation and the Centre de Recherches Math\'{e}matiques, through the Simons-CRM
scholar-in-residence program.
The work of S.B. was supported by NSERC of Canada.  W.D.R. acknowledges the
support of the Flemish Research Fund FWO under grant
G076216N

\end{document}